\documentclass[11pt]{article}

\usepackage{lineno}
\usepackage{amssymb,amsthm,amsmath}
\usepackage{graphicx}
\usepackage[small]{caption}
\usepackage{subcaption}
\usepackage{epsfig}
\usepackage{amsfonts}
\usepackage{complexity}

\usepackage[utf8]{inputenc}
\usepackage{fullpage}

\usepackage{enumerate}
\usepackage{url}
\usepackage{hyperref}

\usepackage[noline,linesnumbered]{algorithm2e}
\SetArgSty{textrm}
\let\oldnl\nl
\newcommand{\nonl}{\renewcommand{\nl}{\let\nl\oldnl}}
\makeatletter
\def\TitleOfAlgo{\@ifnextchar({\@TitleOfAlgoAndComment}{\@TitleOfAlgoNoComment}}
\def\@TitleOfAlgoAndComment(#1)#2{\nonl\hspace*{-1.5em}#2 #1\;}
\def\@TitleOfAlgoNoComment#1{\nonl\hspace*{-1.5em}#1\;}
\makeatother
\DontPrintSemicolon

\newcommand*\patchAmsMathEnvironmentForLineno[1]{
  \expandafter\let\csname old#1\expandafter\endcsname\csname #1\endcsname
  \expandafter\let\csname oldend#1\expandafter\endcsname\csname end#1\endcsname
  \renewenvironment{#1}
  {\linenomath\csname old#1\endcsname}
  {\csname oldend#1\endcsname\endlinenomath}}
  \newcommand*\patchBothAmsMathEnvironmentsForLineno[1]{
  \patchAmsMathEnvironmentForLineno{#1}
  \patchAmsMathEnvironmentForLineno{#1*}}
  \AtBeginDocument{
  \patchBothAmsMathEnvironmentsForLineno{equation}
  \patchBothAmsMathEnvironmentsForLineno{align}
  \patchBothAmsMathEnvironmentsForLineno{flalign}
  \patchBothAmsMathEnvironmentsForLineno{alignat}
  \patchBothAmsMathEnvironmentsForLineno{gather}
  \patchBothAmsMathEnvironmentsForLineno{multline}
}

\newtheorem{theorem}{Theorem}

\newtheorem{corollary}{Corollary}

\newcommand{\etal}{{et~al.}}
\newcommand{\ie}{{i.e.}}
\newcommand{\eg}{{e.g.}}

\newcommand{\opt}{\textsf{OPT}}

\newcommand{\RR}{\mathbb{R}} 
\newcommand{\eps}{\varepsilon}

\def\Prob{{\rm Prob}}
\def\E{{\rm E}}
\def\Var{{\rm Var}}
\def\Cov{{\rm Cov}}

\newcommand{\later}[1]{}
\newcommand{\old}[1]{}

\title{\textsc{Finding Triangles or Independent Sets;\\
and Other Dual Pair Approximations}}

\author{
Adrian Dumitrescu\thanks{%
Algoresearch L.L.C., Milwaukee, WI, USA. 
Email~\texttt{ad.dumitrescu@algoresearch.org}}}

\begin{document}

\maketitle

\begin{abstract}
  We revisit the algorithmic problem of finding a triangle in a graph (\textsc{Triangle Detection}),
  and examine its relation to other problems such as \textsc{3Sum}, \textsc{Independent Set},
  and  \textsc{Graph Coloring}.
We obtain several new algorithms:

\smallskip
(I) A simple randomized algorithm for finding a triangle in a graph.
As an application, we study the range of a conjecture of P\v{a}tra\c{s}cu (2010) regarding
the triangle detection problem.

\smallskip
(II) An algorithm which given a graph $G=(V,E)$ performs one of the
following tasks in $O(m+n)$ (\ie, linear) time:
(i)~compute a $\Omega(1/\sqrt{n})$-approximation of a maximum independent set in $G$ or
(ii)~find a triangle in $G$.
The run-time is faster than that for any previous method for each of these tasks.

\smallskip
(III) An algorithm which given a graph $G=(V,E)$ performs one of the
following tasks in $O(m+n^{3/2})$ time:
(i)~compute an $\sqrt{n}$-approximation for \textsc{Graph Coloring} of $G$ or
(ii)~find a triangle in $G$.
The run-time is faster than that for any previous method for each of these tasks
on dense graphs, with $m =\omega(n^{9/8})$. 

\smallskip
(IV) The second and third results suggest the following broader research direction:
if it is difficult to find (A) or (B) separately, can one find one of the two efficiently?
This motivates the \emph{dual pair} concept we introduce.
We discuss and provide several instances of dual-pair approximation.


\end{abstract}

\section{Introduction} \label{sec:intro}

Consider the problem of deciding whether a given graph $G=(V,E)$ contains a complete subgraph
on $k$ vertices. If the subgraph size $k$ is part of the input, then the problem is $\NP$-complete
(the well-known \textsc{Clique} problem). 
Let $|V|=n, |E|=m$. For every fixed $k$, determining whether a given graph $G=(V,E)$
contains a complete subgraph on $k$ vertices
can be accomplished by a brute-force algorithm running in $O(n^k)$ time.

For $k=3$, deciding whether a graph contains a triangle and finding one if it does
(or counting all triangles in a graph)
can be done in $O(n^\omega)$ time by the algorithm of Itai and Rodeh~\cite{IR78}:
the algorithm computes $M^2$, where $M$ is the graph adjacency matrix.
The existence of entries $ij$ where $M_{ij}=1$ and $M^2_{ij} \geq 1$ indicates the
presence of triangles(s) with edge $ij$. If at least one pair $i,j$ satisfies this condition, then
$G$ contains a triangle. See~\cite[Ch.~10]{Mat10} for a short exposition of this elegant method.
Alternatively, this task can be done in $O(m^{2\omega/(\omega+1)}) =O(m^{1.41})$ time
by the algorithm of Alon, Yuster, and Zwick~\cite{AYZ97}.
For $k=4$, deciding whether a graph contains a $K_4$ and finding one if it does
(or counting all $K_4$'s in a graph) can be done in $O(n^{\omega+1})$ time by the algorithm
of Alon, Yuster, and Zwick~\cite{AYZ97}, and in $O(m^{(\omega+1)/2})=O(m^{1.69})$ time
by the algorithm of Kloks, Kratsch, and M{\"{u}}ller~\cite{KKM00}.

In contrast to the problem of detecting the existence of subgraphs of a certain kind,
the analogous  problem of listing \emph{all} such subgraphs has usually higher complexity.
For example, finding all triangles in a given graph (each triangle appears in the output list)
can be accomplished in $O(m^{3/2})$ time and with $O(m)$ space 
(Itai and Rodeh~\cite{IR78}, Bar-Yehuda and Even~\cite{BE82}).
Chiba and Nishizeki~\cite{CN85} refined the time complexity
in terms of graph arboricity (the minimum number of edge-disjoint
forests into which its edges can be partitioned); their algorithm
lists all triangles in a graph in $O(m \alpha)$ time, where $\alpha$ is the arboricity.
Since there are graphs $G$ with
$\alpha(G)= \Theta(m^{1/2})$, this does not improve the worst-case dependence on $m$
(which, in fact, cannot be improved).
More general, for every fixed $\ell \geq 3$, Chiba and Nishizeki gave an algorithm for
listing all copies of $K_\ell$ in $O(\alpha^{\ell-2} \cdot m)$ time. 

The following variants of the problem of finding triangles in a given undirected graph $G$
can be distinguished:
(i)~the triangle \emph{detection} (or \emph{finding}) problem is that of
  finding a triangle in $G$ or reporting that none exists;
(ii)~the triangle \emph{counting} problem is that of determining the total number 
  of triangles in $G$;
(iii)~the triangle \emph{listing} problem is that of listing all 
triangles in $G$, with each triangle appearing in the output list. 
Any algorithm for listing all triangles can be easily transformed into one for
triangle detection or into one for listing a specified number of triangles (as
called for in Theorem~\ref{thm:patrascu}). 

Our initial motivation in this paper was a conjecture of P\v{a}tra\c{s}cu~\cite{Pa10} regarding triangle detection.
Specifically, he asked whether an $\Omega(m^{4/3-o(1)})$ lower bound holds for this task.

\paragraph{Our results.}

\begin{enumerate} [(i)] \itemsep 1pt

\item Let $0 < \eps \leq 0.5$ and $0<\delta \leq 0.25$ be positive constants. 
Given a graph $G=(V,E)$ with $n$ vertices, $m =\Omega(n^{1+3.82\delta})$ edges,
and $t =\Omega(m^{1+\delta})$ triangles, there is a randomized algorithm that
finds a triangle with high probability in $O\left( n^{\omega(1-\delta)} \right)$ or 
$O \left( \left( m n^{-2\delta} \right)^{\frac{2\omega}{\omega+1}} \right)$
expected time   (Theorem~\ref{thm:detection} in Section~\ref{sec:intro}).
So if an $\Omega(m^{4/3-o(1)})$-lower bound for triangle detection inquired by P\v{a}tra\c{s}cu~\cite{Pa10}
were to hold, then its validity is restricted to a certain range of the graph parameters $m$ and $t$;
in particular, $m$ and $t$ cannot be too large.

\item Let $G=(V,E)$ be a graph with $n$ vertices and $m$ edges.
  We consider data structures for answering independent set queries of the form: Given a subset
  $U \subseteq V$, is $U$ independent?
  We give a simple implementation of a candidate data structure. 

\item Given a graph $G=(V,E)$ with $n$ vertices and $m$ edges,
  an independent set of size $\lceil 2m/n \rceil$          
  or a triangle can be found in $O(m +n)$ time
  (Theorem~\ref{thm:t-or-is} in Section~\ref{sec:t-or-is}).
  Hence, given a graph $G=(V,E)$, one of the following tasks can be performed
  in  $O(m+n)$ time: (i)~compute an $\Omega(1/\sqrt{n})$-approximation of a maximum independent
  set in $G$, or (ii)~find a triangle in $G$.

\item Given a graph $G=(V,E)$ with $n$ vertices and $m$ edges,
  a $\sqrt{n}$-approximation  for \textsc{Graph Coloring} of $G$ or a triangle in $G$
  can be found in $O(m+n^{3/2})$ time
    (Theorem~\ref{thm:t-or-chi} in Section~\ref{sec:t-or-chi}).

\item We discuss and provide several instances of dual-pair approximations besides
  items (iii) and (iv) above:
  Theorems~\ref{thm:eppstein1} and~\ref{thm:eppstein2} in Section~\ref{sec:dual}
  are derived from results of Eppstein~\cite{Epp10}, and Karpinski and Schmied~\cite{KS12}. 
  
\end{enumerate}

\subsection{Preliminaries}

\paragraph{Definitions and notations.} 
Let $G=(V,E)$ be an undirected graph. The \emph{neighborhood} of a vertex $v \in V$ is
the set $N(v) =\{w \ : \  (v,w) \in E\}$ of all adjacent vertices, and its cardinality 
$\deg(v)=|N(v)|$ is called the \emph{degree} of $v$ in $G$.  

A \emph{clique} in a graph $G=(V,E)$ is a subset $C \subseteq V$ of vertices,
each pair of which is connected by an edge in $E$. The \textsc{Clique}
problem is to find a clique of maximum size in $G$. 
An \emph{independent set} of a graph $G=(V,E)$ is a subset $I \subseteq V$ of vertices
such that no two of them are adjacent in $G$. The \textsc{Independent-Set}
problem is to find a maximum-size independent set in $G$. 

Let $\E[X]$ and $\Var[X]$ denote the \emph{expectation} and respectively,
the \emph{variance}, of a random variable~$X$.
If $X$ and $Y$ are random variables,
$\Cov(X,Y) =\E[X Y] - \E[X] \cdot \E[Y]$ is their \emph{covariance}. 
If $E$ is an event in a probability space, $\Prob(E)$ denotes its
probability. Chebyshev's inequality (see for instance~\cite[p.~49]{MU17}) 
is the following: For any $a>0$,
\begin{equation} \label{eq:C}
 \Prob\left(|X -\E[X]| \geq a \right) \leq \frac{\Var[X]}{a^2}.
\end{equation}

Unless specified otherwise, all logarithms are in base~$2$. 

\paragraph{\textsc{3Sum} and triangle detection.}
The \textsc{3Sum} problem is to decide, given a set $A \subset \RR$ of size $n$, whether
there exist $a,b,c \in A$ such that $a + b + c =0$. It was originally conjectured that 
 \textsc{3Sum} requires $\Omega(n^2)$ time on the Real RAM~\cite{BDP08,GP18}.
Gr{\o}nlund and Pettie~\cite{GP18} gave the first algorithms for \textsc{3Sum}
running in subquadratic time and thereby disproved this conjecture.
Specifically, they gave a deterministic algorithm that runs in
$O\left(n^2 (\log \log{n}/\log{n})^{2/3} \right)$ time,
and a randomized algorithm that runs in
$O\left(n^2 (\log \log{n})^2/\log{n} \right)$ expected time.
Gold and Sharir~\cite{GS17} provided a slightly faster deterministic algorithm running in
$O\left(n^2 \log \log{n}/\log{n} \right)$ time.
The revised conjecture is that \textsc{3Sum} requires $\Omega(n^{2-o(1)})$ time
on the Real RAM~\cite{GP18}. Obtaining a specified time bound for a problem is
said to be ``\textsc{3Sum}-hard'' if doing so would violate the revised conjecture above. 

The \textsc{3Sum} problem and its variants are quite relevant in the field of algorithm
complexity. Although the \textsc{3Sum} problem itself does not seem to have practical
applications, it has enjoyed wide interest due to numerous problems that can be reduced
from it.
Thus, lower bounds on \textsc{3Sum} imply lower bounds on several problems
in computational geometry, dynamic graph algorithms, triangle listing,
and others; see~\cite{GO12,GP18} for an enumeration of such problems. 

Recall that testing whether a graph contains a triangle can be done in
$O(m^{2\omega/(\omega+1)}) =O(m^{1.41})$ time by the algorithm of Alon, Yuster,
and Zwick~\cite{AYZ97};
and so if $\omega=2$ this task can be accomplished in $O(m^{4/3})$ time.
Exploiting the connection with \textsc{3Sum}, P\v{a}tra\c{s}cu obtained the
following conditional lower bound (recall the meaning of ``\textsc{3Sum}-hard''):

\begin{theorem} {\rm \cite{Pa10}} \label{thm:patrascu}
  In a graph with $m$ edges, listing $m$ triangles in $O(m^{4/3-\eps})$ time is
  \textsc{3Sum}-hard.
\end{theorem}

The problem of listing $m$ triangles examined by P\v{a}tra\c{s}cu is a bit artificial,
and one can argue that the ``real'' problem in this area is \textsc{Triangle Detection}. 
Indeed, P\v{a}tra\c{s}cu also remarked that it would be very interesting to extend
the $\Omega(m^{4/3-o(1)})$ lower bound of Theorem~\ref{thm:patrascu} to the computationally
easier task of triangle detection. While the questioned extension refers to the worst-case scenario
(\ie, graph), it is natural to investigate its limits across the entire range of~$t$. 
At one end of the range, it certainly does not apply to graphs with arboricity $O(1)$,
where the number of triangles $t$ is small, and triangle detection takes $O(m)$ time
(by the result of~\cite{CN85}, as mentioned earlier). 
Shifting to the other end of the range, we have dense graphs with a superlinear number of
triangles.  We prove the following algorithmic result for triangle detection
that relies on the input graph having a super-linear number of edges and triangles;
it shows that one cannot expect the conjecture to hold at this end either. 

\begin{theorem} \label{thm:detection}
Let $0 < \eps \leq 0.5$ and $0<\delta \leq 0.25$ be positive constants. 
Given a graph $G=(V,E)$ with $n$ vertices, $m =\Omega(n^{1+3.82\delta})$ edges,
and $t =\Omega(m^{1+\delta})$ triangles, there is a randomized algorithm that
finds  a triangle in $G$ with high probability.
  The expected running time of the algorithm is
\[ O\left( n^{\omega(1-\delta)} \right) \text{ or }
O \left( \left( m n^{-2\delta} \right)^{\frac{2\omega}{\omega+1}} \right). \]
\end{theorem}

The running time of our algorithm beats the one of the algorithm of Alon~\etal~for the class of graphs
described in Theorem~\ref{thm:detection} for every $\delta$. Indeed, we have
\[ m n^{-2\delta} = o(m), \text{ thus } 
\left( m n^{-2\delta} \right)^{\frac{2\omega}{\omega+1}} = o\left( m^{\frac{2\omega}{\omega+1}} \right). \]
Moreover, the running time of our algorithm beats the one in Theorem~\ref{thm:patrascu}
for certain values of $\delta$, even with the \emph{current} fastest matrix multiplication algorithm.
Indeed, for $\delta \geq 0.055$ we have $\left( m n^{-2\delta} \right)^{1.41} = o(m^{4/3-1/1000})$. 
The strong point of the new algorithm featured in Theorem~\ref{thm:detection} is its simplicity.

\section{Triangle detection in graphs with a few triangles} \label{sec:detection}

In this section we outline a randomized algorithm for triangle detection and thereby
prove Theorem~\ref{thm:detection}. Whereas it is very simple and natural, it does not appear
to have been considered and analyzed previously. 
The algorithm calls for triangle detection (via counting) in the graph induced by a small sample
of vertices, as given by the algorithm of Itai and Rodeh~\cite{IR78}
or by the algorithm of Alon, Yuster, and Zwick~\cite{AYZ97}.
Recall that the two algorithms run in $O(n^\omega)=O(n^{2.372})$ time and
$O(m^{2\omega/(\omega+1)}) =O(m^{1.41})$ time, respectively, on input graphs
with $n$ vertices and $m$ edges.
We refer to any of these algorithms as \textsc{Exact-Count}$(G)$. 

\smallskip
\begin{algorithm}[H]
  \DontPrintSemicolon
  \TitleOfAlgo{\textsc{Triangle-Detection}$(G)$}
  \KwIn{an undirected graph $G=(V,E)$ and $0<\delta \leq 0.25$}
  \KwOut{a triangle in $G$ (if found)}
  Extract a random vertex sample $U \subseteq V$ by retaining each vertex
  with probability $p=n^{-\delta}$\;
  Let $Z$ be the number of triangles in the induced graph $G[U]$ obtained
  by the algorithm \textsc{Exact-Count}\;
  If $Z>0$ return a triangle and halt
\end{algorithm}

\paragraph{Analysis.}
An analysis built on similar principles can be found in~\cite[Ch.~6.5]{MU17}; however, the specifics here
are quite different. Let $t$ denote the number of triangles in $G$. 
Let $X=|U|$ (the size of $U$). 
Let $Y$ and $Z$ be the random variables denoting the number of edges and triangles
that survive in $G[U]$.
By the linearity of expectation, we have
\begin{align}
  \E[X] &=n p = n^{1-\delta}, \label{eq:X} \\
  \E[Y] &= m p^2 = m n^{-2\delta}, \label{eq:Y} \\
  \E[Z] &= t p^3 = t n^{-3\delta}.
\end{align}
The third equality yields $ \E[Z n^{3\delta}] = n^{3\delta} \cdot \E[Z] = t$, as intended.
The expectations $\E[X]$, $\E[Y]$ and $\E[Z]$ are used for determining the expected running times
$\E\left[ X^\omega \right]$ and $\E\left[ Y^{\frac{2\omega}{\omega+1}} \right]$
and the probability of fail (\ie, no triangle is found).
Deviation inequalities are used throughout the proof.

We first estimate the running time. Let $\mu_X=\E[X]=n^{1-\delta}$.
By Chernoff bounds on the sum of Bernoulli trials~\cite[Thm.~4.4]{MU17}
and taking~\eqref{eq:X} into account, we have
\begin{equation}
  \Prob(X \geq 6\mu_X) = \Prob(X \geq 6n^{1-\delta}) \leq
  2^{-6n^{1-\delta}} \leq 2^{-n^{3/4}}. 
\end{equation}
This yields
\begin{align*}
  \E\left[ X^\omega \right] &\leq  \Prob(X \leq 6\mu_X) \cdot (6\mu_X)^\omega
  + \Prob(X \geq 6\mu_X) \cdot n^\omega \\
&\leq \left( 6n^{1-\delta} \right)^\omega + 2^{-n^{3/4}} n^\omega 
= O\left( n^{\omega(1-\delta)} \right).
\end{align*}

For $i=1,\ldots,m$, define the indicator random variables $Y_i$ by
\begin{align*}
Y_i &=
\left\{ \begin{array}{ll}
1 & \text{if edge } i \text{ appears in } G[U] \\
0 & \text{else}. \end{array} \right.
\end{align*}

We have $\E[Y_i]=p^2$. 
Then $Y= \sum_{i=1}^m Y_i$. Note that the variables $Y_i$ are (in general) not independent.
The variance of $Y$ is given by
\begin{equation} \label{eq:var-1}
\Var[Y] = \Var \left[\sum_{i=1}^m Y_i \right] = \sum_{i=1}^m \Var[Y_i]
+ 2 \sum_{i<j} \Cov(Y_i,Y_j).
\end{equation}

We have
\[ \Var[Y_i] = \E[Y_i] - (\E[Y_i])^2 \leq  \E[Y_i]=p^2, \]
and
\begin{align*}
  \Cov(Y_i,Y_j) &=\E[Y_i Y_j] - \E[Y_i] \E[Y_j] \\
  &=\left\{ \begin{array}{lll}
0 & \text{ if the two edges have no vertex in common}, \\
p^3 - p^4 & \text{ if the two edges have a vertex in common}.
\end{array} \right.
\end{align*}

The number of pairs of edges having exactly one common vertex is (see for instance~\cite{Ca98}):
$ \sum_{v \in V}^n {\deg(v) \choose 2} \leq mn$. 
This implies
\[ \sum_{i<j} \Cov(Y_i,Y_j) \leq mn p^3, \text{ thus }
\Var[Y] \leq mp^2 + 2mn p^3. \]

Let $\mu_Y=\E[Y]=m n^{-2\delta}$. 
Taking~\eqref{eq:Y} into account and applying Chebyshev's inequality \eqref{eq:C}
to $Y$ yields
\[ \Prob \left( \left| Y - \mu_Y\right| \geq \mu_Y \right)
  \leq \frac{\Var[Y]}{\mu_Y^2} \\
  \leq \frac{m p^2 + 2mn p^3}{m^2 p^4} =
  \frac{1 + 2np}{m p^2} = O\left( \frac{n^{1+\delta}}{m}\right) .
\]
Consequently we obtain 
\begin{align*}
  \E\left[ Y^{\frac{2\omega}{\omega+1}} \right] &\leq
  \Prob(Y \leq 2\mu_Y) \cdot (2\mu_Y)^{\frac{2\omega}{\omega+1}}
  + \Prob(Y \geq 2\mu_Y) \cdot m^{\frac{2\omega}{\omega+1}} \\
  &= O\left( \left(m n^{-2\delta} \right)^{\frac{2\omega}{\omega+1}} 
  +  \frac{n^{1+\delta}}{m} \cdot m^{\frac{2\omega}{\omega+1}} \right) 
= O \left( \left( m n^{-2\delta} \right)^{\frac{2\omega}{\omega+1}} \right),
\end{align*}
where in the last step we have used the inequality
$n^{1+ \left(1 + \frac{4\omega}{\omega+1}\right)\delta}
\leq n^{1+3.82\delta} =O(m)$.
It follows that the expected running time is 
\[ O\left( n^{\omega(1-\delta)} \right) \text{ or }
O \left( \left( m n^{-2\delta} \right)^{\frac{2\omega}{\omega+1}} \right). \]

We next analyze the probability of fail (\ie, no triangle is found).
Arbitrarily label the triangles in $G$ by $1,2,\ldots,t$. 
For $i=1,\ldots,t$, define the indicator random variables $Z_i$ by
\begin{align*}
Z_i &=
\left\{ \begin{array}{ll}
1 & \text{if triangle } i \text{ appears in } G[U] \\
0 & \text{else}. \end{array} \right.
\end{align*}

We have $\E[Z_i]=p^3$. 
Then $Z= \sum_{i=1}^t Z_i$. Note that the variables $Z_i$ are (in general) not independent.
The variance of $Z$ is given by
\begin{equation} \label{eq:var-2}
\Var[Z] = \Var \left[\sum_{i=1}^t Z_i \right] = \sum_{i=1}^t \Var[Z_i]
+ 2 \sum_{i<j} \Cov(Z_i,Z_j).
\end{equation}

We have
\[ \Var[Z_i] = \E[Z_i] - (\E[Z_i])^2 \leq  \E[Z_i]=p^3, \]
and
\begin{align*}
  \Cov(Z_i,Z_j) &=\E[Z_i Z_j] - \E[Z_i] \E[Z_j] \\
  &=\left\{ \begin{array}{lll}
0 & \text{ if the two triangles have no vertex in common}, \\
p^5 - p^6 & \text{ if the two triangles have a single vertex in common}, \\
p^4 - p^6 & \text{ if the two triangles have two vertices in common}.
\end{array} \right.
\end{align*}

The number of pairs of triangles having one or two vertices in common
are at most $3tm$ and $3tn$, respectively. 
By substituting these expressions into~\eqref{eq:var-2} we obtain
an upper bound on the variance of $Z$:
\begin{equation} \label{eq:var-3}
  \Var[Z] \leq t p^3 + 6tm p^5 + 6tn p^4.
\end{equation}

We assumed that $t=\Omega\left(m^{1+\delta}\right)$, thus by the second moment method,
see, \eg, \cite[Chap.~6.5]{MU17}, we have
\begin{align*}
 \Prob(Z=0) &\leq \frac{\Var[Z]}{(\E[Z])^2} \\
  &\leq \frac{(t p^3 + 6tm p^5 + 6tn p^4)}{t^2 p^6} \\
  &\leq \frac{6}{t} \left(n^{3\delta} + m n^{\delta} + n^{1+2\delta} \right)\\
  &= O\left(\frac{m n^{\delta} + n^{1+2\delta}}{\eps^2 t} \right) 
  = O\left(n^{-3.82 \delta^2}\right).
\end{align*}   

If a nonzero count is returned by the algorithm,
$G[U]$ contains $Z>0$ triangles, and
one of these triangles can be also identified within the same time.
This concludes the proof of Theorem~\ref{thm:detection}. 
\qed

\section{Data structures for independent set queries}
\label{sec:indep-set-queries}

Let $G=(V,E)$ be a graph with $n$ vertices and $m$ edges.
  We consider data structures for answering independent set queries of the form: Given a subset
  $U \subseteq V$, is $U$ independent?
  We present two simple implementations of a candidate data structure
  (if desired, they can be merged into one). 

\smallskip
The first one only uses the adjacency list data structure of $G$. 
The corresponding running time is $O(|U| +m)$.

\smallskip
\begin{algorithm}[H]
  \DontPrintSemicolon
  \TitleOfAlgo{\textsc{Independent-Set-Query-1}$(U)$}
  \KwIn{a subset $U \subseteq V$, where $G=(V,E)$}
  Mark and link in a list the entries of $V$ that appear in $U$ (remove the marks in the end before halting)\;
  \ForEach{edge $(u,v) \in E$}{
    Determine whether $u \in U$ and $v \in U$\;
    If $u \in U$ and $v \in U$, output that $U$ is not independent and halt\;
    Else continue\; }
    If no edge in $E$ is spanned by $U$, output that $U$ is independent and halt\;
\end{algorithm}

\smallskip
The second one uses the adjacency matrix $M$ of $G$. 
The corresponding running time is $O(|U|^2)= O(n^2)$.

\smallskip
\begin{algorithm}[H]
  \DontPrintSemicolon
  \TitleOfAlgo{\textsc{Independent-Set-Query-2}$(U)$}
  \KwIn{a subset $U \subseteq V$, where $G=(V,E)$}
  \ForEach{pair $(u,v) \in U^2$}{
If $M_{uv}=1$, output that $U$ is not independent and halt\;
    Else continue\; }
  If no pair in $U^2$ is an edge in $E$, output that $U$ is independent and halt\;
\end{algorithm}

Observe that \textsc{Independent-Set-Query-1} is efficient for sparse graphs
and large query sets, whereas \textsc{Independent-Set-Query-2}  is efficient
for small query sets. By choosing the best alternative in each case, one can answer queries
in $O\left(\min(m,  |U|^2 \right)$ time; note that the crossover is when $|U| =\Theta(m^{1/2})$.

\section{Triangles or independent sets}   \label{sec:t-or-is}

As the dual of \textsc{Clique}, the  \textsc{Independent-Set} problem is known to be
$\NP$-complete~\cite{GJ79}.
Any clique $C$ in $G$ is an independent set of the same size in $\overline{G}$,
the \emph{complement} of $G$ and vice versa. As such, any approximation 
algorithm for one of these problems can be converted to an approximation 
algorithm with the same approximation ratio for the other problem by simply
running it on the complement graph~\cite[Ch.~10.2]{WS11}.   
The \textsc{Clique} problem is also hard to approximate. First,
there is no constant approximation algorithm
for \textsc{Clique} unless $\P=\NP$~\cite[p.~421]{WS11}. Further, according to a result of
Zuckerman~\cite{Zu07}, for every positive constant $\eps>0$,  it is $\NP$-hard to approximate
\textsc{Clique} to within $n^{1-\eps}$. The same results hold for \textsc{Independent-Set}.
The best approximation algorithm known achieves an approximation ratio of
$\Omega\left(\log^3{n}/(n (\log \log {n})^2)\right)$~\cite{Fei04}.

According to a celebrated theorem of Tur\'{a}n~\cite{Tu41}, every graph of order $n$
and average degree $\delta$ contains an independent set of size at least $n/(\delta+1)$. 
As such, sparse graphs have large independent sets.  
A constructive proof of Tur\'{a}n's Theorem given by Erd\H{o}s yields a linear-time
greedy algorithm---included below---for finding an independent set of this size;
see for instance~\cite[p.~118]{Ho97}.
If $G$ is dense, \ie, it has $\Theta(n^2)$ edges, then $\delta = \Omega(n)$,
the independent set size guaranteed by the  $n/(\delta+1)$ bound is only $O(1)$ and
computing a constant-factor approximation is ruled out unless $\P=\NP$. 

\smallskip
\begin{algorithm}[H]
  \DontPrintSemicolon
  \TitleOfAlgo{\textsc{Independent-Set}$(G)$}
  \KwIn{an undirected graph $G=(V,E)$}
  Set $I \gets \emptyset$\;
  If $G$ is empty then stop; otherwise choose a vertex $v$
  of minimum degree in the current graph\;
  Add $v$ to $I$, delete $v$ and all its neighbors (along with all edges incident to at least
  one of these vertices) from $G$ and go to Step 2\;
\end{algorithm}
%

Shifting now to the triangle detection problem, it is easily solvable in polynomial time
by a brute force algorithm running in $O(n^3)$ time, or by the faster algorithms
in $O(n^\omega)$ time~\cite{IR78}
or in $O(m^{2\omega/(\omega+1)}) =O(m^{1.41})$ time~\cite{AYZ97}.
Due to its importance, the triangle detection problem along with its many variants
(listing all triangles, or listing only a prescribed number, or counting the triangles)
has received lots of attention starting in the 1980s and more recently in the perspective
of developing lower bounds for dynamic problems~\cite{AWY18,KPP16,Pa10}.

Consider the problem pair
$\langle \textsc{Independent-Set}, \textsc{Triangle Detection} \rangle$. 
Oddly enough, if one allows an algorithm the freedom to decide which problem to solve,
\ie, whether to find a large independent set \emph{or} a triangle, then there exists a very simple
and fast deterministic algorithm. The results are summarized in Theorem~\ref{thm:t-or-is}
and Corollary~\ref{cor:t-or-is}. Note that the most efficient algorithm for each of the two problems
runs in super-linear time. 

\begin{theorem} \label{thm:t-or-is}
  Given a graph $G=(V,E)$ with $n$ vertices and $m$ edges,
  an independent set of size $\lceil 2m/n \rceil$          
  or a triangle can be found in $O(m +n)$ time.
\end{theorem}

The size of the independent set found by the algorithm 
in Theorem~\ref{thm:t-or-is} grows inversely proportionally with the size of
the independent set found by the greedy algorithm. 
In particular,
\begin{itemize} \itemsep 1pt
\item [-] If $m=\Theta(n^2)$, an independent set of size $\Omega(n)$ or a triangle can be found
  in $O(m +n)$ time. If an independent set is returned, it is notably a constant-factor
  approximation; indeed, $\opt \leq n$.
  In contrast, the greedy algorithm is guaranteed only a set of size $O(1)$!
\item [-] If $m=\Theta(d n)$, an independent set of size $\Omega(d)$ or a triangle can be found
  in $O(m +n)$ time. In contrast, the greedy algorithm finds a set of size $\Omega(n/d)$. 
\end{itemize}

\begin{proof} (of Theorem~\ref{thm:t-or-is}). 
Let $v$ be a vertex of maximum degree in $G$ and let $\Gamma(v)$ be its neighborhood. 
Since the average degree in $G$ is $2m/n$, we have $|\Gamma(v)| \geq \lceil 2m/n \rceil$.
Arbitrarily retain a subset $U \subseteq \Gamma(v)$ of this size: $|U|= \lceil 2m/n \rceil$.
If $U$ is an independent set we are done; in the other case we are also done since
a triangle incident to $v$ has been found: if, say, $x,y \in U$ and $xy \in E$, then
$\langle v,x,y \rangle$ is a triangle in $G$.
Using either implementation of the data structure in Section~\ref{sec:indep-set-queries},
the independent set test takes $O(m+n)$ time. Indeed,
\begin{equation} \label{eq:fast1}
  O\left(|U|^2 + m +n\right) = O\left(m^2/n^2 + m +n\right) = O(m +n),
\end{equation}
and
\begin{equation} \label{eq:fast2}
  O\left(|U| + m +n\right) = O(m +n),
\end{equation}
as required.
\end{proof}

\begin{corollary} \label{cor:t-or-is}
Given a graph $G=(V,E)$, one of the following tasks can be performed in  $O(m+n)$ time:
(i)~compute an $\Omega(1/\sqrt{n})$-approximation of a maximum independent set in $G$,
or (ii)~find a triangle in $G$.
\end{corollary}
\begin{proof}
Let $G$ be a graph of order $n$. Since a largest independent set has size at most $n$,
finding an independent set of size $\Omega(\sqrt{n})$ achieves the first objective.
If the maximum degree is at most $\sqrt{n}$ the greedy algorithm \textsc{Independent-Set}$(G)$
yields an independent set of size about $\sqrt{n}$ thereby achieving the first objective.
If the maximum degree is at least $\sqrt{n}$, let $v$ be a vertex of maximum degree.
Arbitrarily retain a subset $U \subset \Gamma(v)$ of this size: $|U|= \lceil \sqrt{n} \rceil$.
If $U$ is an independent set the first objective has been achieved.
In the other case a triangle incident to $v$ has been found and the second objective has
been achieved. 
\end{proof}

\paragraph{Remarks.}
\textsc{Independent-Set} remains $\NP$-complete for triangle-free graphs~\cite{Pol74};
see also~\cite[p.~194--195]{GJ79}.
Let $G$ be a triangle-free graph of order $n$.
Repeatedly removing independent sets of size about $\sqrt{n}$ yields an efficient algorithm
for coloring $G$ with $O(\sqrt{n})$ colors. Consequently, the chromatic number is
$\chi(G)= O(\sqrt{n})$. A slightly larger independent set is implied from results of Ajtai, 
Koml\'{o}s and Szemer\'{e}di~\cite{AKS80}, who showed that if $G$ is a triangle-free
graph of order $n$, then $G$ contains an independent set of size $\Omega(\sqrt{n \log{n}})$.
Consequently, $\chi(G)= O(\sqrt{n/\log{n}})$. Apart from the constant factor, this
bound is the best possible (by a celebrated result of Kim~\cite{Kim95}). 
See also~\cite{GT00} for a perspective on these results.

\section{Triangles or colorings}   \label{sec:t-or-chi}

A \emph{coloring} of an undirected graph $G=(V,E)$ is a partition of vertices into color classes so that
no edge joins two vertices in the same class. The  \textsc{Graph Coloring} problem is that of computing
such a partition using as few colors as possible. This minimum number is the \emph{chromatic number}
$\chi(G)$ of $G$. The problem is known to be $\NP$-hard~\cite{GJ79},
and the best approximation ratio known is $O\left(n (\log \log{n})^2/\log^3{n} \right)$~\cite{Ha93}.
It is conjectured~\cite{Ha93} that the best possible approximation guarantee for \textsc{Graph Coloring} is
$O\left(n/\log^c{n}\right)$, for some constant $c \geq 3$. 

The argument in the proof of Theorem~\ref{thm:t-or-is} is similar to 
arguments used in coloring triangle-free graphs. See for instance, \cite[Prop.~8.1.18]{We21}.
Let $G$ be a triangle-free graph of order $n$.
Repeatedly removing independent sets of size about $\sqrt{n}$ yields an efficient algorithm
for coloring $G$ with $O(\sqrt{n})$ colors. Consequently, the chromatic number is
$\chi(G)= O(\sqrt{n})$. We adapt this algorithm to obtain a good approximation for \textsc{Graph Coloring}
or to find a triangle in $G$ efficiently. 

\begin{theorem} \label{thm:t-or-chi}
  Given a graph $G=(V,E)$ with $n$ vertices and $m$ edges,
 a $\sqrt{n}$-approximation for \textsc{Graph Coloring} of $G$ or a triangle in $G$
  can be found in $O(m+n^{3/2})$ time.
\end{theorem}
\begin{proof}
  As long as $G$ has a vertex $v$ with at least $\lfloor \sqrt{n} \rfloor$ neighbors not yet colored,
  arbitrarily retain a subset $U \subseteq \Gamma(v)$ of uncolored vertices of this size:
  $|U|= \lfloor \sqrt{n} \rfloor$. If $U$ is an independent set, use one new color on these vertices
  and repeat. Otherwise, a triangle in $G$ has been found, and the algorithm halts. 
  Since $G$ has $n$ vertices, this first phase uses at most $\sqrt{n}$ colors.
  Afterwards, the subgraph $G'$ induced by the remaining vertices has maximum degree less than
  $\lfloor \sqrt{n} \rfloor$. In the second phase, arbitrarily order the remaining vertices and use
  the greedy algorithm, see, \eg, \cite[p.~147]{Bol98}, to color $G'$ with at most  $\sqrt{n}$
  additional colors. The total number of colors used is at most $2 \sqrt{n}$.

 Since the independent set test in one iteration takes $O(n)$ time, and there are   
 at most $\sqrt{n}$ iterations,  Phase I takes $O(m+n^{3/2})$ time.
 Phase II takes $O(m+n)$ time. 
 Consequently, the run-time of the algorithm is $O(m+n^{3/2})$.

It is easy to see that the coloring algorithm has ratio at most $\sqrt{n}$ on every nonempty graph $G$. 
Indeed, $\chi(G) \geq 2$, and $2 \sqrt{n}/2 = \sqrt{n}$. 
\end{proof}

Assume that $m = \omega(n^{9/8})$.
Note that $m+n^{3/2} = o(n^w)$ unless $\omega=2$, but our algorithm is combinatorial and much simpler
than any algorithm for triangle detection based on matrix multiplication. Moreover, presently it is only known
that $\omega<2.372$. 
Note also that $m+n^{3/2} = o(m^{2\omega/(\omega+1)}) =o(m^{1.41})$, regardless of the value of $\omega$:
indeed, $n^{3/2} \ll m^{2\omega/(\omega+1)}$ if $m \gg n^{\frac{3(\omega+1)}{4 \omega}}$, which holds by the
assumption.
For the current state of the art in matrix multiplication, we have $\omega<2.372$, thus our algorithm is faster
already for $m = \Omega(n^{1.066})$.

\section{Dual pair approximations}   \label{sec:dual}

Theorem~\ref{thm:t-or-is} suggests the following broader research direction:
If it is difficult to find (A) or (B) separately, can one find one of the two efficiently?

In the context of dealing with hard problems, Vassilevska~\etal~\cite{VWW06} 
proposed a hybrid method, \ie, the use of \emph{hybrid algorithms}.
Specifically, the authors demonstrated $\NP$-hard problems that admit a hybrid algorithm 
where a given instance can either be solved exactly in subexponential time, or be
approximated in polynomial time but with an approximation ratio that is better
than the inapproximability threshold of the problem, assuming $\P \neq \NP$. 
This question is somewhat analogous to some recent approaches in
fine grain complexity studies, \eg, \cite{AWY18},
where one would like to obtain conditional lower bounds that rely on the hypothesized
hardness of \emph{at least one} of several problems. 

Informally, a pair of computational problems $\langle A,B \rangle$ with the same type of inputs
(\eg, graphs), is called a \emph{dual pair} if neither $A$, nor $B$, admits an efficient algorithm
but there is an efficient algorithm to solve either $A$ or $B$ on the \emph{same} instance.

Let \textsc{$\rho$-Indep.-Set} denote the problem of computing a $\rho$-approximation
of a maximum independent set in $G$, where $G=(V,E)$, $|V|=n$, $|E|=m$. 
We showed that \linebreak
$\langle \textsc{$\Omega(1/\sqrt{n})$-Indep.-Set}, \textsc{Triangle Detection} \rangle$ 
is a dual pair (Cor.~\ref{cor:t-or-is}).
Further, Let \textsc{$\rho$-Graph Coloring} denote the problem of computing a $\rho$-approximation
for \textsc{Graph Coloring}, where $G=(V,E)$, $|V|=n$, $|E|=m$. 
We showed that
$\langle \textsc{$\sqrt{n})$-Graph Coloring}, \textsc{Triangle Detection} \rangle$ 
is a dual pair (Thm.~\ref{thm:t-or-chi}) in dense graphs, with $m =\omega(n^{9/8})$. 

Several other results in this direction, due to Eppstein~\cite{Epp10},
can be refined as follows.
For an undirected graph $G=(V,E)$, one can obtain a \textsc{(1,2)-TSP} instance
in a canonical way by using unit weights for every edge in $E$ and weight $2$ for every non-edge
(\ie, element of $\overline{E}$).
\textsc{(1,2)-TSP} was first studied in~\cite{PY93}. 
The best known polynomial approximation for \textsc{(1,2)-TSP} has ratio $8/7$~\cite{BK06,AMP18}. 
Let \textsc{$\rho$-(1,2)-TSP} denote the problem of computing a $\rho$-approximation
for \textsc{(1,2)-TSP} on the weighted graph defined from $G$ as above.

Specifically, Eppstein~\cite[Thm.~2.2]{Epp10} has shown (via DFS) that 
$ \langle \textsc {$(1+\eps)$-(1,2)-TSP}, \textsc{$\eps$-Indep.-Set} \rangle $ 
is a dual pair.
In principle, $\eps$ can be arbitrary small, which is undesirable,
  however, setting $\eps=1/535$ yields an approximation beyond the current
  inapproximability ratio for each of the two problems. Indeed, Karpinski and Schmied~\cite{KS12,KLS15}
  proved that \textsc{(1,2)-TSP} is $\NP$-hard to approximate with a factor less than $535/534$.
  Consequently, this result can be strengthened as follows.

\begin{theorem} \label{thm:eppstein1}
  For an undirected graph $G=(V,E)$, and its corresponding  \textsc{(1,2)-TSP} instance,
  a $1/535$ approximation for \textsc{Independent-Set} or a  $536/535$ approximation for
   \textsc{(1,2)-TSP} can be obtained in linear time.
\end{theorem} 

  Similarly, for a directed graph $G=(V,E)$, one can obtain a \textsc{(1,2)-ATSP} (asymmetric \textsc{TSP})
  instance in the same way.
Let \textsc{$\rho$-(1,2)-ATSP} denote the problem of computing a $\rho$-approximation
for \textsc{(1,2)-ATSP} on the weighted directed graph defined from $G$ as above.
For a directed graph $G=(V,E)$,  \textsc{MAX-ACY-IND-SG} is
the problem of finding the largest cardinality of a subset of vertices whose induced subgraph is acyclic.
It is known~\cite{LY93} that this problem  is $\NP$-hard to approximate with a factor of $2^{\log^c n}$,
for any constant $0<c<1/2$.

 Analogous to Tur\'{a}n's result mentioned in Section~\ref{sec:t-or-is},
  every directed graph of order $n$ and average out-degree $\delta^+$ contains an induced subgraph
  of size at least $n/(\delta^++1)$, see~\cite{BD06}.
As such, sparse directed graphs have large induced acyclic subgraphs.  
A constructive proof follows along the lines of~\cite{BD06}. 

For directed graphs, Eppstein~\cite[Sec.~4]{Epp10} has shown (again, via DFS) that \linebreak
$ \langle \textsc {$(1+\eps)$-(1,2)-ATSP}, \textsc{$\eps$-{MAX-ACY-IND-SG}} \rangle $ 
is a dual pair.
Again, in principle $\eps$ can be arbitrary small, 
  however, setting $\eps=1/207$ yields an approximation beyond the current
  inapproximability ratio for each of the two problems. Indeed, Karpinski and Schmied~\cite{KS12,KLS15}
  proved that \textsc{(1,2)-ATSP} is $\NP$-hard to approximate with a factor less than $207/206$.
  Consequently, this result can be strengthened as follows.

\begin{theorem} \label{thm:eppstein2}
  For a directed graph $G=(V,E)$, and its corresponding  \textsc{(1,2)-ATSP} instance,
a $1/207$ approximation for \textsc{MAX-ACY-IND-SG} or a $208/207$ approximation for
   \textsc{(1,2)-ATSP} can be obtained in linear time.
\end{theorem}

\end{document}